\documentclass[a4paper,USenglish,numberwithinsect]{lipics-v2018}

\usepackage{microtype}

\graphicspath{{./fig/}}

\title{How Bad is the Freedom to  \textsc{Flood-It}?}

\author{R\'{e}my~Belmonte}
{The University of Electro-Communications, Chofu, Tokyo, Japan}
{remy.belmonte@uec.ac.jp}
{}
{}
\author{Mehdi~Khosravian~Ghadikolaei}
{Université Paris-Dauphine, PSL Research University, CNRS, UMR,\\ LAMSADE, 75016 Paris, France}
{mehdi.khosravian-ghadikolaei@dauphine.fr}
{}
{}
\author{Masashi~Kiyomi}
{Yokohama City University, Yokohama, Japan}
{masashi@yokohama-cu.ac.jp}
{}
{}
\author{Michael~Lampis}
{Université Paris-Dauphine, PSL Research University, CNRS, UMR,\\ LAMSADE, 75016 Paris, France}
{michail.lampis@dauphine.fr}
{}
{}
\author{Yota~Otachi}
{Kumamoto University, Kumamoto, Japan}
{otachi@cs.kumamoto-u.ac.jp}
{0000-0002-0087-853X}
{}

\authorrunning{R. Belmonte, M. Khosravian Ghadikolaei, M. Kiyomi, M. Lampis, and Y. Otachi}

\Copyright{R\'{e}my Belmonte, Mehdi Khosravian Ghadikolaei, Masashi Kiyomi, Michael Lampis, and Yota Otachi}

\subjclass{, }


\funding{This work is partially supported by JSPS and MAEDI under the Japan-France Integrated Action Program (SAKURA) Project GRAPA 38593YJ.}

\theoremstyle{plain}
\newtheorem{proposition}[theorem]{Proposition}

\newcommand{\figref}[1]{Figure~\ref{#1}}

\newcommand{\col}{\ensuremath\mathbf{col}}
\newcommand{\comp}{\ensuremath\mathbf{Comp}}

\newcommand{\free}{\textsc{Free-Flood-It}}
\newcommand{\fixed}{\textsc{Fixed-Flood-It}}
\newcommand{\optfree}{\textrm{OPT}_{\textrm{Free}}}
\newcommand{\optfixed}{\textrm{OPT}_{\textrm{Fixed}}}

\newcommand{\sfree}{\textsc{Subset-Free-Flood-It}}
\newcommand{\sfixed}{\textsc{Subset-Fixed-Flood-It}}
\newcommand{\optsfree}{\textrm{OPT}_{\textrm{S-Free}}}
\newcommand{\optsfixed}{\textrm{OPT}_{\textrm{S-Fixed}}}

\newcommand{\setc}{\textsc{Set Cover}}
\newcommand{\Msc}{\textsc{Multi-Colored Set Cover}} 
\newcommand{\msc}{\textsc{MCSC}}

\newcommand{\vc}{\mathsf{vc}}
\newcommand{\nd}{\mathsf{nd}}
\newcommand{\mw}{\mathsf{mw}}
\newcommand{\cw}{\mathsf{cw}}

\newtheorem*{lemma*}{Lemma}
\newtheorem*{theorem*}{Theorem}
\newtheorem*{corollary*}{Corollary}
\newtheorem*{claim*}{Claim}

\usepackage[color=green!20]{todonotes}

\usepackage{algorithm}
\usepackage[noend]{algpseudocode}

\usepackage{color}
\definecolor{darkred}{rgb}{0.8,0,0}

\EventEditors{John Q. Open and Joan R. Acces}
\EventNoEds{2}
\EventLongTitle{42nd Conference on Very Important Topics (CVIT 2016)}
\EventShortTitle{CVIT 2016}
\EventAcronym{CVIT}
\EventYear{2016}
\EventDate{December 24--27, 2016}
\EventLocation{Little Whinging, United Kingdom}
\EventLogo{}
\SeriesVolume{42}
\ArticleNo{23}

\hideLIPIcs
\nolinenumbers

\begin{document}

\maketitle

\begin{abstract}
{\fixed} and {\free} are combinatorial problems on graphs that generalize a
very popular puzzle called \textit{Flood-It}.  Both problems consist of
recoloring moves whose goal is to produce a monochromatic (``flooded'') graph
as quickly as possible. Their difference is that in {\free} the player has the
additional freedom of choosing the vertex to play in each move.  In this paper,
we investigate how this freedom affects the complexity of the problem.  It
turns out that the freedom is \emph{bad} in some sense.  We show that some
cases trivially solvable for {\fixed} become intractable for {\free}.
We also show that some tractable cases for {\fixed} are still tractable for
{\free} but need considerably more involved arguments.
We finally present some combinatorial properties connecting or separating the
two problems.  In particular, we show that the length of an optimal solution
for {\fixed} is always at most twice that of {\free}, and this is tight.
\end{abstract}


\section{Introduction}
\label{sec:intro}
\emph{Flood-It} is a popular puzzle, originally released as a computer game in
2006 by LabPixies (see \cite{CliffordJMS12}). In this game, the player is
presented with (what can be thought of as) a vertex-colored grid graph, with a
designated special \emph{pivot} vertex, usually the top-left corner of the
grid.  In each move, the player has the right to change the color of all
vertices contained in the same monochromatic component as the pivot to a
different color of her choosing. Doing this judiciously gradually increases the
size of the pivot's monochromatic component, until the whole graph is
\emph{flooded} with one color.  The goal is to achieve this flooding with the
minimum number of moves. See \figref{fig:example-fixed} for an example.
\begin{figure}[htb]
  \centering
  \includegraphics[scale=0.8]{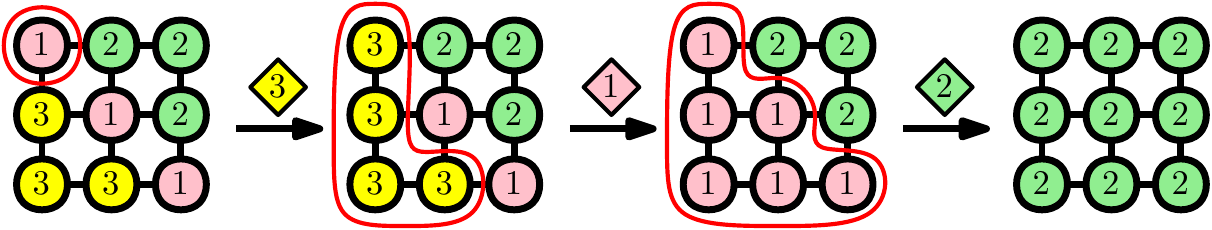}
  \caption{A flooding sequence on a $3 \times 3$ grid.
    Each move in this example changes the color of the top-left monochromatic component.
    Under such a restriction, the depicted sequence is shortest.}
  \label{fig:example-fixed}
\end{figure}

Following the description above, \emph{Flood-It} immediately gives rise to a
natural optimization problem: given a vertex-colored graph, determine the
shortest sequence of flooding moves that wins the game.  This problem has been
extensively studied in the last few years (e.g.
\cite{LagoutteNT14,MeeksS12,MeeksS14,MeeksS13,FukuiOUUU12,FellowsSPS15,SouzaPS14,FellowsPRSS17,MeeksV15arxiv,HonKLLW15arxiv};
a more detailed summary of known results is given below), both because of the
game's popularity (and addictiveness!), but also because the computational
complexity questions associated with this problem have turned out to be
surprisingly deep, and the problem has turned out to be surprisingly
intractable.

The goal of this paper is to add to our understanding of this interesting,
puzzle-inspired, optimization problem, by taking a closer look at the
importance of the \emph{pivot vertex}. As explained above, the classical
version of the game only allows the player to change the color of a special
vertex and its component and has been studied under the name \fixed~\cite{MeeksS12,MeeksS14,MeeksS13}
(or \textsc{Flood-It} in some papers~\cite{CliffordJMS12,SouzaPS14,FellowsSPS15,FellowsPRSS17,HonKLLW15arxiv}).
However, it is extremely natural to also consider a version
where the player is also allowed to play a different vertex of her choosing in
each turn. This has also been well-studied under the name 
\free~\cite{CliffordJMS12,LagoutteNT14,MeeksS12,MeeksS14,MeeksS13,FellowsSPS15,SouzaPS14}.
See \figref{fig:example-free}.
\begin{figure}[htb]
  \centering
  \includegraphics[scale=0.8]{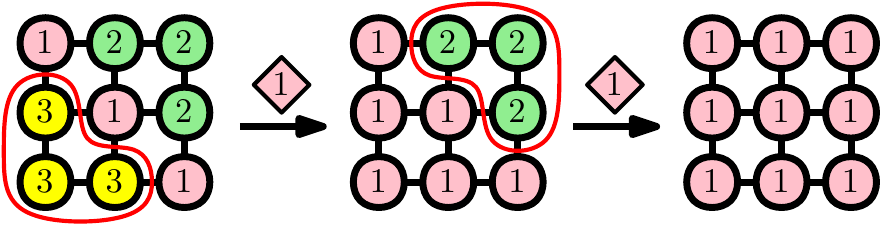}
  \caption{A flooding sequence with no restriction on selected monochromatic components.
    This is shorter than then one in \figref{fig:example-fixed}.}
  \label{fig:example-free}
\end{figure}

Since both versions of this problem have been studied before, the question of
the impact of the pivot vertex on the problem's structure has (at least
implicitly) been considered. Intuitively, one would expect \free\ to be a
harder problem; after all, the player has to choose a color to play \emph{and}
a vertex to play it on, and is hence presented with a larger set of possible
moves. The state of the art seems to confirm this intuition, as
only some of the positive algorithmic results known for \fixed\ are known also
for \free, while there do exist some isolated cases where \fixed\ is tractable
and \free\ is hard, for example co-comparability graphs
\cite{FleischerW12,FukuiOUUU12} and grids of height $2$
\cite{CliffordJMS12,MeeksS13}.  Nevertheless, these results do not completely
pinpoint the added complexity brought by the task of selecting a vertex to
play, as the mentioned algorithms for \fixed\ are already non-trivial, and
hence the jump in complexity is likely to be the result of \emph{the
combination} of the tasks of picking a color and a vertex.  More broadly,
\cite{FellowsSPS15} presented a generic reduction from \fixed\ to \free\ that
preserves a number of crucial parameters (number of colors, optimal value,
etc.) and gives convincing evidence that \free\ is always at least as hard as
\fixed, but not necessarily harder.

\subparagraph*{Our Results} We investigate the complexity of \free, mostly from
the point of view of parameterized complexity,\footnote{For readers unfamiliar
with the basic notions of this field, we refer to standard textbooks
\cite{CyganFKLMPPS15,FlumG06}.} as well as the impact on the combinatorics of
the game of allowing moves outside the pivot.

Our first result is to show that \free\ is W[2]-hard parameterized by the
number of moves in an optimal solution. We recall that for \fixed\ this
parameterization is trivially fixed-parameter tractable: when a player has only
$k$ moves available, then we can safely assume that the graph uses at most
(roughly) $k$ colors, hence one can easily consider all possible solutions in
FPT time. The interest of our result is, therefore, to demonstrate that the
task of deciding which vertex to play next is sufficient to make \free\
significantly harder than \fixed. Indeed, the W[2]-hardness reduction we give,
implies also that \free\ is not solvable in $n^{o(k)}$ time under the ETH. This
tightly matches the complexity of a trivial algorithm which considers all
possible vertices and colors to be played. This is the first concrete example showing a
case where \fixed\ is essentially trivial, but \free\ is intractable.

Motivated by this negative result we consider several other parameterizations
of the problem. We show that {\free} is fixed-parameter tractable when
parameterized by the number of possible moves and the clique-width.  This
result is tight in the sense that the problem is hard when parameterized by
only one of these parameters.  It also implies
the fixed-parameter tractability of the problem parameterized by the number of
colors and the modular-width.
In a similar vein, we present a polynomial kernel when \free\ is parameterized
by the input graph's neighborhood diversity and number of colors.  An analogous
result was shown for \fixed\ in \cite{FellowsPRSS17}, but because of the
freedom to select vertices, several of the tricks used there do not apply to
\free, and our proofs are slightly more involved.  Our previously mentioned
reduction also implies that \free\ does not admit a polynomial kernel
parameterized by vertex cover, under standard assumptions. This result was also
shown for \fixed\ in \cite{FellowsPRSS17}, but it does not follow immediately
for \free, as the reduction of \cite{FellowsSPS15} does not preserve the
graph's vertex cover.

Motivated by the above results, which indicate that the complexity of the
problem can be seriously affected if one allows non-pivot moves, we also study
some more purely combinatorial questions with algorithmic applications. The
main question we pose here is the following. It is obvious that for all
instances the optimal number of moves for \free\ is upper-bounded by the
optimal number of moves for \fixed\ (since the player has strictly more
choices), and it is not hard to construct instances where \fixed\ needs
strictly more moves. Can we bound the optimal number of \fixed\ moves needed as
a function of the optimal number of \fixed\ moves? Somewhat surprisingly, this
extremely natural question does not seem to have been explicitly considered in
the literature before. Here, we completely resolve it by showing that the two
optimal values cannot be more than a factor of $2$ apart, and constructing a
family of simple instances where they are exactly a factor of $2$ apart. As an
immediate application, this gives a $2$-approximation for \free\ for every case
where \fixed\ is known to be tractable. 

We also consider the problem's monotonicity: \fixed\ has the nice property that
even an adversary that selects a single bad move cannot increase the optimal
(that is, in the worst case a bad move is a wasted move). We construct minimal
examples which show that \free\ does not have this nice monotonicity property,
even for extremely simple graphs, that is, making a bad move may not only waste
a move but also make the instance strictly worse.
Such a difference was not explicitly stated in the literature,
while the monotonicity of {\fixed} was seem to be known or at least assumed.
The only result we are aware of is the monotonicity of {\free} on paths
shown by Meeks and Scott~\cite{MeeksS12}.

\paragraph*{Known results} In 2009, the NP-hardness of {\fixed} with six colors
was sketched by Elad Verbin as a comment to a blog post by Sariel
Har-Peled~\cite{Verbin2009blog-comment}.  Independently to the blog comment,
Clifford et al.~\cite{CliffordJMS12} and Fleischer and
Woeginger~\cite{FleischerW12} started investigations of the complexity of the
problem, and published the conference versions of their papers at FUN 2010.
Here we mostly summarize some of the known results on {\free}.  For more
complete lists of previous result, see
e.g.~\cite{FukuiOUUU12,LagoutteNT14,FellowsPRSS17}.

{\free} is NP-hard if the number of colors is at least 3~\cite{CliffordJMS12} even for trees with only one vertex of 
degree more than 2~\cite{LagoutteNT14,FellowsSPS15},
while it is polynomial-time solvable for general graphs 
if the number of colors is at most 2~\cite{CliffordJMS12,MeeksS12,LagoutteNT14}.
Moreover, it is NP-hard even for height-3 grids with four colors~\cite{MeeksS12}.
Note that this result implies that {\free} with a constant number colors is NP-hard
even for graphs of bounded bandwidth.
If the number of colors is unbounded, then
it is NP-hard for height-2 grids~\cite{MeeksS13},
trees of radius 2~\cite{FellowsSPS15}, and,
proper interval graphs and caterpillars~\cite{FukuiOUUU12}.
Also, it is known that there is no constant-factor approximation with a factor independent of the number of colors
unless $\text{P} = \text{NP}$~\cite{CliffordJMS12}.

There are a few positive results on {\free}.
Meeks and Scott~\cite{MeeksS14} showed that
every colored graph has a spanning tree with the same coloring
such that the minimum number of moves coincides in the graph and the spanning tree.
Using this property, they showed that if a graph has only a polynomial number of vertex subsets that induce connected subgraphs,
then {\free} (and {\fixed}) on the graph can be solved in polynomial time.
This in particular implies the polynomial-time solvability on subdivisions of a fixed graph.
It is also known that {\free} for interval graphs and split graphs 
is fixed-parameter tractable when parameterized by the number of colors~\cite{FukuiOUUU12}.


\section{Preliminaries}
\label{sec:pre}

For a positive integer $k$, we use $[k]$ to denote the set $\{1,\ldots,k\}$.
Given a graph $G=(V,E)$, a coloring function $\col\colon V\to [c_{\max}]$, where
$c_{\max}$ is a positive integer, and $u\in V$, we denote by $\comp(\col, u)$
the maximal set of vertices $S$ such that for all $v\in S$, $\col(u)=\col(v)$
and there exists a path from $u$ to $v$ such that for all its internal vertices
$w$ we have $\col(w)=\col(u)$. In other words, $\comp(\col,u)$ is the
monochromatic connected component that contains $u$ under the coloring function
$\col$.

Given $G, \col$, a \emph{move} is defined as a pair $(u,i)$ where $u\in V$,
$i\in[c_{\max}]$.  The \emph{result} of the move $(u,c)$ is a new coloring
function $\col'$ defined as follows: $\col'(v) = c$ for all $v\in
\comp(\col,u)$; $\col'(v) = \col(v)$ for all other vertices. In words, a move
consists of changing the color of $u$, and of all vertices in the same
monochromatic component as $u$, to $c$. Given the above definition we can also
define the result of a sequence of moves $(u_1,c_1), (u_2,c_2),\ldots, (u_k,
c_k)$ on a colored graph with initial coloring function $\col_0$ in the natural
way, that is, for each $i\in [k]$, $\col_i$ is the result of move $(u_i,c_i)$ on
$\col_{i-1}$.

The \free\ problem is defined as follows: given a graph $G=(V,E)$, an integer
$k$, and an initial coloring function $\col_0$, decide if there exists a
sequence of $k$ moves $(u_1,c_1), (u_2,c_2),\ldots, (u_k,c_k)$ such that the
result $\col_k$ obtained by applying this sequence of moves on $\col_0$ is a
constant function (that is, $\forall u,v\in V$ we have $\col_k(u)=\col_k(v)$).

In the \fixed\ problem we are given the same input as in the \free\ problem, as
well as a designated vertex $p\in V$ (the pivot). The question is again if
there exists a sequence of moves such that $\col_k$ is monochromatic, with the
added constraint that we must have $u_i=p$ for all $i\in [k]$.

We denote by $\optfree(G,\col), \optfixed(G,\col,p)$ the minimum $k$ such that
for the input $(G, \col)$ (or $(G,\col, p)$ respectively) the \free\ problem
(respectively the \fixed\ problem) admits a solution.

\subsection{Graph parameters}

\textbf{Vertex cover number:}
A set $S \subseteq V$ is a \emph{vertex cover} of a graph $G= (V,E)$ if each edge in $E$ has at least one end point in $S$.
The minimum size of a vertex cover of a graph is its \emph{vertex cover number}.
By $\vc(G)$, we denote the vertex cover number of $G$.

\textbf{Neighborhood diversity:}
Let $G = (V,E)$ be a graph.
For each vertex $v \in V$, we denote the neighborhood of $v$ by $N_{G}(v)$.
The \emph{closed neighborhood} of $v$ is the set $N_{G}[v] := \{v\} \cup N_{G}(v)$.
We omit the subscript $G$ when the underlying graph is clear from the context.
Two vertices $u,v \in V$ are \emph{true twins} in $G$ if $N[u] = N[v]$
and are \emph{false twins} in $G$ if $N(u) = N(v)$.
Two vertices are \emph{twins} if they are true twins or false twins.
Note that true twins are adjacent and false twins are not.
The \emph{neighborhood diversity} of $G$, denoted $\nd(G)$, is the minimum number $k$
such that $V$ can be partitioned into $k$ sets of twin vertices.
It is known that $\nd(G) \le 2^{\vc(G)} + \vc(G)$ for every graph $G$~\cite{Lampis12}. 
Given a graph, its neighborhood diversity and the corresponding partition into
sets of twins can be computed in polynomial time~\cite{Lampis12}; in fact,
using fast modular decomposition algorithms, the neighborhood diversity of a
graph can be computed in linear time \cite{McConnellS99,TedderCHP08}.

\textbf{Modular-width:}
Let $H$ be a graph with $k \ge 2$ vertices $v_{1}, \dots, v_{k}$,
 and let $H_{1}, \dots, H_{k}$ be $k$ graphs.
The \emph{substitution} $H(H_{1}, \dots, H_{k})$ of the vertices of $H$ by $H_{1}, \dots, H_{k}$ is the graph
with the vertex set $\bigcup_{1 \le i \le k} V(H_{i})$
and the edge set $\bigcup_{1 \le i \le k} V(E_{i}) \cup \{\{u, w\} \mid u \in V(H_{i}), w \in V(H_{j}), \{v_{i}, v_{j}\} \in E(H)\}$.
For each $i$, the set $V(H_{i})$ is called a \emph{module} of $H(H_{1}, \dots, H_{k})$.
The \emph{modular-width} of $G$, denoted $\mw(G)$, is defined recursively as follows:
\begin{itemize}
  \item If $G$ has only one vertex, then $\mw(G) = 1$.
  \item If $G$ is the disjoint union of graphs $G_{1}, \dots, G_{h}$, then $\mw(G) = \max_{1 \le i \le h} \mw(G_{i})$.

  \item If $G$ is a connected graph with two or more vertices, then
  \[
    \mw(G) = \min_{H,H_{1},\dots,H_{|V(H)|}}\max\{|V(H)|, \mw(H_{1}), \dots, \mw(H_{|V(H)|})\},
  \]
  where the minimum is taken over all tuples of graphs $(H,H_{1}, \dots, H_{k})$ such that $G = H(H_{1}, \dots, H_{|V(H)|})$.
\end{itemize}
A recursive substitution structure giving the modular-width can be computed in linear-time~\cite{McConnellS99,TedderCHP08}.
It is known that $\mw(G) \le \nd(G)$ for every graph $G$~\cite{GajarskyLO13}.

\textbf{Clique-width:}
A \emph{$k$-expression} is a rooted binary tree such that
\begin{itemize}
  \setlength{\itemsep}{0pt}
  \item each leaf has label $\circ_{i}$ for some $i \in \{1,\dots,k\}$,
  \item each non-leaf node with two children has label $\cup$, and
  \item each non-leaf node with only one child has label $\rho_{i,j}$ or $\eta_{i,j}$ ($i, j \in \{1,\dots,k\}$, $i \ne j$).
\end{itemize}
Each node in a $k$-expression represents a vertex-labeled graph as follows:
\begin{itemize}
  \setlength{\itemsep}{0pt}
  \item a $\circ_{i}$-node represents a graph with one $i$-vertex;
  \item a $\cup$-node represents the disjoint union of the labeled graphs represented by its children;
  \item a $\rho_{i,j}$-node represents the labeled graph obtained from the one represented by its child by replacing
  the labels of the $i$-vertices with $j$;
  \item an $\eta_{i,j}$-node represents the labeled graph obtained from the one represented by its child
  by adding all possible edges between the $i$-vertices and the $j$-vertices.
\end{itemize}
A $k$-expression represents the graph represented by its root.
The \emph{clique-width} of a graph $G$, denoted by $\cw(G)$,
is the minimum integer $k$ such that there is a $k$-expression representing a graph isomorphic to $G$.
From their definitions, $\cw(G) \le \mw(G)$ holds for every graph $G$.

\figref{fig:width-parameters} shows relationships among the graph parameters introduced above
together with the well-known treewidth and pathwidth
(see \cite{CyganFKLMPPS15} for definitions of these two parameters).
\begin{figure}[htb]
  \centering
  \includegraphics[scale=0.9]{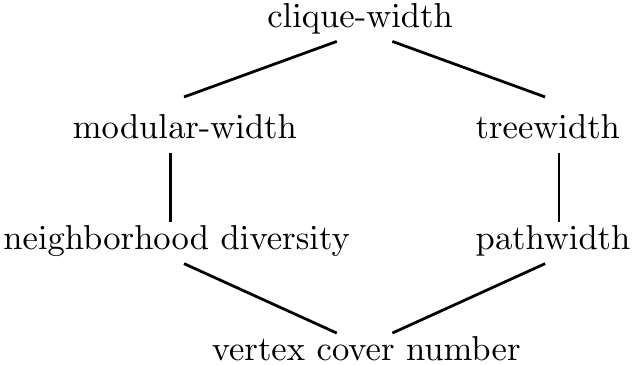}
  \caption{Graph parameters. Each segment implies that the one above is more general than the one below.
  For example, bounded modular-width implies bounded clique-width but not vice versa.}
  \label{fig:width-parameters}
\end{figure}


\section{W[2]-hardness of \free}
\label{sec:whard}

\tikzstyle{vertex}=[circle, draw, inner sep=2.5pt, minimum width=4pt, minimum size=0.1cm]
\tikzstyle{vertex1}=[circle, draw, inner sep=0pt, minimum width=1pt, minimum size=0cm]
\usetikzlibrary{decorations,decorations.pathmorphing,decorations.pathreplacing,fit,positioning}

The main result of this section is that \free\ is W[2]-hard when parameterized
by the minimum length of any valid solution (the natural parameter). The proof
consists of a reduction from \setc, a canonical W[2]-complete problem. 

Before presenting the construction, we recall two basic observations by Meeks and Vu~\cite{MeeksV15arxiv},
both of which rest on the fact that any single move can (at most) eliminate a
single color from the graph, and this can only happen if a color induces a
single component.

\begin{lemma}
[\cite{MeeksV15arxiv}]
\label{lem:col-1}

For any graph $G=(V,E)$, and coloring function $\col$ that uses $c_{\max}$
distinct colors, we have $\optfree(G,\col) \ge c_{\max}-1$. 

\end{lemma}

%

\begin{lemma}
[\cite{MeeksV15arxiv}]
\label{lem:col} 

For any graph $G=(V,E)$, and coloring function $\col$ that uses $c_{\max}$
distinct colors, such that for all $c\in [c_{\max}]$, $G[\col^{-1}(c)]$
is a disconnected graph, we have $\optfree(G,\col) \ge c_{\max}$.

\end{lemma}


The proof of Theorem \ref{thm:whard} relies on a reduction from a special form
of \setc, which we call \Msc\ (\msc\ for short). \msc\ is defined as follows:

\begin{definition}

In \Msc\ (\msc) we are given as input a set of elements $R$ and $k$ collections
of subsets of $R$, $\mathcal{S}_1,\ldots, \mathcal{S}_k$. We are asked if there
exist $k$ sets $S_1,\ldots, S_k$ such that for all $i\in [k]$,
$S_i\in\mathcal{S}_i$, and $\cup_{i\in[k]} S_i = R$.

\end{definition}

Observe that \msc\ is just a version of \setc\ where the collection of sets is
given to us pre-partitioned into $k$ parts and we are asked to select one set
from each part to form a set cover of the universe. It is not hard to see that
any \setc\ instance $(\mathcal{S},R)$ where we are asked if there exists a set
cover of size $k$ can easily be transformed to an equivalent \msc\ instance
simply by setting $\mathcal{S}_i=\mathcal{S}$ for all $i\in[k]$, since the
definition of \msc\ does not require that the sub-collections $\mathcal{S}_i$
be disjoint. We conclude that known hardness results for \setc\ immediately
transfer to \msc, and in particular \msc\ is W[2]-hard when parameterized by
$k$.

\subsubsection*{Construction}

We are now ready to describe our reduction which, given a \msc\ instance with
universe $R$ and $k$ collections of sets $\mathcal{S}_i, i\in[k]$, produces an
equivalent instance of \free, that is, a graph $G=(V,E)$ and a coloring
function $\col$ on $V$.  We construct this graph as follows:

\begin{itemize}

\item for every set $S\in \mathcal{S}_i$, construct a vertex in $V$.  The set
of vertices in $V$ corresponding to sets of $\mathcal{S}_i$ is denoted by $I_i$
and $\col(v)=i$ for each $v\in I_i$.  $I_1\cup ...\cup I_k$ induces an
independent set colored $\{1,...,k\}$.

\item for each $i\in [k]$, construct $3k$ new vertices, denoted by $L_i$ and
connect all of them to all vertices of $I_i$ such that $L_i\cup I_i$ induces a
complete bipartite graph of size $3k\times |I_i|$. Then set $\col(v)=k+1$ for
each $v\in L_i$, for all $i\in[k]$.

\item for each vertex $v\in L_i$ for $1\leq i\leq k$, construct $k$ new leaf
vertices connected to $v$ with distinct colors $1,...,k$.

\item for each element $e\in R$, construct a vertex $e$. For each $S\in\mathcal{S}_i$ such that $e\in S$ we connect $e$ to the vertex of $I_i$ that represents $S$.

\item add a special vertex $u$ with $\col(u)=k+1$ which is connected it to all vertices in $I_i$ for $i\in[k]$.

\end{itemize} 

An illustration of $G$ is shown in Fig.\ref{Fig: free-flood-it}. In the following we will show that $(G,\col)$ as an instance of \free\ is solvable with at most $2k$ moves if and only if the given \msc\ instance has a set cover of size $k$ which contains one set of each $\mathcal{S}_i$.

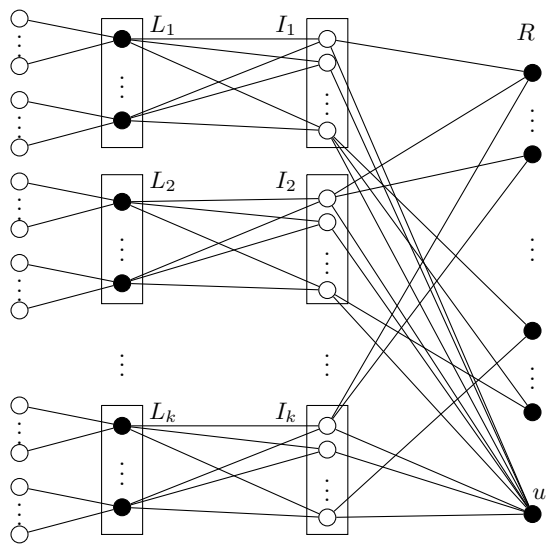
\begin{figure}
\centering
\begin{tikzpicture}[scale=0.9, transform shape]

\node[vertex] (k11) at (0,0) {};
\node[vertex,below of=k11,node distance=0.35cm](k12){};
\node[below of=k12, node distance=0.5 cm] (k1dots) {$\vdots$};
\node[vertex,below of=k1dots,node distance=0.5cm](k1c){};

\node[vertex,below of=k1c,node distance=1cm](k21){};
\node[vertex,below of=k21,node distance=0.35cm](k22){};
\node[below of=k22, node distance=0.5 cm] (k2dots) {$\vdots$};
\node[vertex,below of=k2dots,node distance=0.5cm](k2c){};

\node[below of=k2c, node distance=1 cm] (kdots) {$\vdots$};

\node[vertex,below of=kdots,node distance=1cm](kk1){};
\node[vertex,below of=kk1,node distance=0.35cm](kk2){};
\node[below of=kk2, node distance=0.5 cm] (kkdots) {$\vdots$};
\node[vertex,below of=kkdots,node distance=0.5cm](kkc){};

\node[vertex,fill=black!100] (l11) at (-3,0) {};
\node[below of=l11, node distance=0.6 cm] (l1dots) {$\vdots$};
\node[vertex,fill=black!100,below of=l1dots,node distance=0.6cm](l13k){};

\node[vertex,fill=black!100,below of=l13k,node distance=1.2 cm](l21){};
\node[below of=l21, node distance=0.6 cm] (l2dots) {$\vdots$};
\node[vertex,fill=black!100,below of=l2dots,node distance=0.6cm](l23k){};

\node[below of=l23k, node distance=1.1 cm] (ldots) {$\vdots$};

\node[vertex,fill=black!100,below of=ldots,node distance=1cm](lk1){};
\node[below of=lk1, node distance=0.6 cm] (lkdots) {$\vdots$};
\node[vertex,fill=black!100,below of=lkdots,node distance=0.6cm](lk3k){};


\node[vertex,fill=black!100] (r11) at (3,-0.5) {};
\node[below of=r11, node distance=0.6 cm] (r1dots) {$\vdots$};
\node[vertex,fill=black!100,below of=r1dots,node distance=0.6cm](r13k){};


\node[below of=r13k, node distance=1.3 cm] (rdots) {$\vdots$};

\node[vertex,fill=black!100,below of=rdots,node distance=1.3cm](rm1){};
\node[below of=rm1, node distance=0.6 cm] (rmdots) {$\vdots$};
\node[vertex,fill=black!100,below of=rmdots,node distance=0.6cm](rm3k){};

\node[vertex,fill=black!100,below of=rm3k,node distance=1.5cm](u){};


\node[vertex] (p1l11) at (-4.5,0.3) {};
\node[below of=p1l11, node distance=0.3 cm] (pl11dots) {$\vdots$};
\node[vertex,below of=pl11dots,node distance=0.4cm](pkl11){};

\node[vertex,below of=pkl11,node distance=0.5cm](p1l13k){};
\node[below of=p1l13k, node distance=0.3 cm] (pl1kdots) {$\vdots$};
\node[vertex,below of=pl1kdots,node distance=0.4cm](pkl13k){};

\node[vertex] (p1l21) at (-4.5,-2.1) {};
\node[below of=p1l21, node distance=0.3 cm] (pl21dots) {$\vdots$};
\node[vertex,below of=pl21dots,node distance=0.4cm](pkl21){};

\node[vertex,below of=pkl21,node distance=0.5cm](p1l23k){};
\node[below of=p1l23k, node distance=0.3 cm] (pl2kdots) {$\vdots$};
\node[vertex,below of=pl2kdots,node distance=0.4cm](pkl23k){};

\node[vertex] (p1lk1) at (-4.5,-5.4) {};
\node[below of=p1lk1, node distance=0.3 cm] (plk1dots) {$\vdots$};
\node[vertex,below of=plk1dots,node distance=0.4cm](pklk1){};

\node[vertex,below of=pklk1,node distance=0.5cm](p1lk3k){};
\node[below of=p1lk3k, node distance=0.3 cm] (plkkdots) {$\vdots$};
\node[vertex,below of=plkkdots,node distance=0.4cm](pklk3k){};

\draw (p1l11)--(l11)--(pkl11);
\draw (p1l13k)--(l13k)--(pkl13k);

\draw (p1l21)--(l21)--(pkl21);
\draw (p1l23k)--(l23k)--(pkl23k);

\draw (p1lk1)--(lk1)--(pklk1);
\draw (p1lk3k)--(lk3k)--(pklk3k);

\draw (l11)--(k11)--(l13k);
\draw (l11)--(k12)--(l13k);
\draw (l11)--(k1c)--(l13k);

\draw (l21)--(k21)--(l23k);
\draw (l21)--(k22)--(l23k);
\draw (l21)--(k2c)--(l23k);

\draw (lk1)--(kk1)--(lk3k);
\draw (lk1)--(kk2)--(lk3k);
\draw (lk1)--(kkc)--(lk3k);







\draw (r11)--(k11);
\draw (r11)--(k21);
\draw (r11)--(kk1);
\draw (r13k)--(k21);
\draw (r13k)--(kk1);

\draw (rm1)--(k1c);
\draw (rm1)--(kkc);
\draw (rm3k)--(k1c);
\draw (rm3k)--(k2c);

\draw (u)--(k11);
\draw (u)--(k12);
\draw (u)--(k1c);
\draw (u)--(k21);
\draw (u)--(k22);
\draw (u)--(k2c);
\draw (u)--(kk1);
\draw (u)--(kk2);
\draw (u)--(kkc);



\draw (-0.3,0.3)--(0.3,0.3)--(0.3,-1.6)--(-0.3,-1.6)--(-0.3,0.3);
\draw (-0.3,-2)--(0.3,-2)--(0.3,-3.9)--(-0.3,-3.9)--(-0.3,-2);
\draw (-0.3,-5.4)--(0.3,-5.4)--(0.3,-7.3)--(-0.3,-7.3)--(-0.3,-5.4);

\draw (-3.3,0.3)--(-2.7,0.3)--(-2.7,-1.6)--(-3.3,-1.6)--(-3.3,0.3);
\draw (-3.3,-2)--(-2.7,-2)--(-2.7,-3.9)--(-3.3,-3.9)--(-3.3,-2);
\draw (-3.3,-5.4)--(-2.7,-5.4)--(-2.7,-7.3)--(-3.3,-7.3)--(-3.3,-5.4);



\node () at (-0.6,0.2) {$I_1$};
\node () at (-0.6,-2.1) {$I_2$};
\node () at (-0.6,-5.5) {$I_k$};

\node () at (-2.4,0.2) {$L_1$};
\node () at (-2.4,-2.1) {$L_2$};
\node () at (-2.4,-5.5) {$L_k$};

\node () at (2.9,0.1) {$R$};
\node () at (3.1,-6.7) {$u$};

\end{tikzpicture}
\caption{The graph $G=(V,E)$ of \free\ constructed from the given \msc\ instance. All the vertices in each $I_i$ have color $i$ and all black vertices have color $k+1$. Boxes containing black vertices have size $3k$. Also each vertex in $L_i$ has $k$ neighbors with degree 1 colored $1,...,k$.}\label{Fig: free-flood-it}
\end{figure}

\begin{lemma}\label{lem: forward dir}

If $(\mathcal{S}_1,\ldots,\mathcal{S}_k,R)$ is a YES instance of \msc\ then $\optfree(G,\col)\le 2k$. 

\end{lemma}

\begin{proof}

Suppose that there is a solution $S_1,\ldots,S_k$ of the given \msc\ instance,
with $S_i\in\mathcal{S}_i$, for $i\in[k]$ and $\cup_{i\in[k]}S_i = R$. Recall
that for each $S_i$ there is a vertex in $I_i$ in the constructed graph
representing $S_i$. Our first $k$ moves consist of changing the color of each
of these $k$ vertices to $k+1$ in some arbitrary order.

Observe that in the graph resulting after these $k$ moves the vertices with
color $k+1$ form a single connected component: because $\cup S_i$ is a set
cover, all vertices of $R$ have a neighbor with color $k+1$; all vertices with
color $k+1$ in some $I_i$ are in the same component as $u$; and all vertices of
$\cup_{i\in[k]} L_i$ are connected to one of the vertices we played.
Furthermore, observe that this component dominates the graph: all remaining
vertices of $\cup I_i$, as well as all leaves attached to vertices of
$\cup_{i\in[k]} L_i$ are dominated by the vertices of $\cup_{i\in[k]} L_i$.
Hence, we can select an arbitrary vertex with color $k+1$, say $u$, and cycle
through the colors $1,\ldots,k$ on this vertex to make the graph monochromatic.
\end{proof}

Now we establish the converse  of Lemma \ref{lem: forward dir}.

\begin{lemma}\label{lem: backward dir}

If $\optfree(G,\col)\le 2k$, then $(\mathcal{S}_1,\ldots,\mathcal{S}_k,R)$ is a
YES instance of \msc.

\end{lemma}

\begin{proof}

Suppose that there exists a sequence of at most $2k$ moves solving $(G,\col)$.
We can assume without loss of generality that the sequence has length exactly
$2k$, since performing a move on a monochromatic graph keeps the graph
monochromatic.  Let $(u_1,c_1),\ldots,(u_{2k},c_{2k})$ be a solution, let
$\col_0=\col$, and let $\col_i$ denote the coloring of $G$ obtained after the
first $i$ moves.  The key observation that we will rely on is the following:

\medskip 

$(i)$ For all $i\in[k]$, there exist $j\in[k], v\in I_i$ such that
$\col_j(v)=k+1$.

\medskip

In other words, we claim that for each group $I_i$ there exists a vertex that
received color $k+1$ at some point during the first $k$ moves. Before
proceeding, let us prove this claim. Suppose for contradiction that the claim
is false. Then, there exists a group $I_i$ such that no vertex in that group
has color $k+1$ in any of the colorings $\col_0,\ldots,\col_k$. We now consider
the vertices of $L_i$ and their attached leaves. Since $L_i$ contains $3k>k+2$
vertices, there exist two vertices $v_1,v_2$ of $L_i$ such that
$\{u_1,\ldots,u_k\}$ contains neither $v_1,v_2$, nor any of their attached
leaves. In other words, there exist two vertices of $L_i$ on which the winning
sequence does not change colors by playing them or their private neighborhood
directly.  However, since $v_1,v_2$ only have neighbors in $I_1$ (except for
their attached leaves), and no vertex of $I_1$ received color $k+1$, we
conclude that $\col_k(v_1)=\col_k(v_2)=k+1$, that is, the colors of these two
vertices have remained unchanged, and the same is true for their attached
leaves. Consider now the graph $G$ with coloring $\col_k$: we observe that this
coloring uses $k+1$ distinct colors, and that each color induces a disconnected
graph. This is true for colors $1,\ldots,k$ because of the leaves attached to
$v_1,v_2$, and true of color $k+1$ because of $v_1,v_2$ and the fact that no
vertex of $I_i$ has color $k+1$. We conclude that $\optfree(G,\col_k)\ge k+1$
by Lemma \ref{lem:col}, which is a contradiction, because the whole sequence
has length $2k$.

\medskip

Because of claim (i) we can now conclude that for all $i\in[k]$ there exists a
$j\in[k]$ such that $\col_{j-1}(u_j)=i$. In other words, for each color $i$
there exists a move among the first $k$ moves of the solution that played a
vertex which at that point had color $i$. To see that this is true consider
again for contradiction the case that for some $i\in[k]$ this statement does
not hold: this implies that vertices with color $i$ in $\col_0$ still have
color $i$ in $\col_1,\ldots,\col_k$, which means that no vertex of $I_i$ has
received color $k+1$ in the first $k$ moves, contradicting (i).

As a result of the above, we therefore claim that for all $j\in[k]$, we have
$\col_{j-1}(u_j)\neq k+1$. In other words, we claim that none of the first $k$
moves changes the color of a vertex that at that point had color $k+1$. This is
because, as argued, for each of the other $k$ colors, there is a move among the
first $k$ moves that changes a vertex of that color. We therefore conclude that
for all vertices $v$ for which $\col_0(v)=k+1$ we have $\col_j(v)=k+1$ for all
$j\in[k]$. In addition, because in $\col_0$ all colors induce independent sets,
each of the first $k$ moves changes the color of a single vertex. Because of
claim (i), this means that for each $i\in[k]$ one of the first $k$ moves
changes the color of a single vertex from $I_i$ to $k+1$. We select the
corresponding set of $\mathcal{S}_i$ in our \msc\ solution.

We now observe that, since all vertices of $\cup_{i\in[k]} L_i$ retain color
$k+1$ throughout the first $k$ moves, $\col_k$ is a coloring function that uses
$k+1$ distinct colors, and colors $1,\ldots,k$ induce disconnected graphs
(because of the leaves attached to the vertices of each $L_i$). Thanks to Lemma
\ref{lem:col}, this means that $\col_k^{-1}(k+1)$ must induce a connected
graph.  Hence, all vertices of $R$ have a neighbor with color $k+1$ in
$\col_k$, which must be one of the $k$ vertices played in the first $k$ moves;
hence the corresponding element is dominated by our solution and we have a
valid set cover selecting one set from each $\mathcal{S}_i$.  \end{proof}

We are now ready to combine Lemmas \ref{lem: forward dir} and \ref{lem:
backward dir} to obtain the main result of this section. 

\begin{theorem}\label{thm:whard} \free\ is W[2]-hard parameterized by
$\optfree$, that is, parameterized by the length of the optimal solution.
Furthermore, if there is an algorithm that decides if a \free\ instance has a
solution of length $k$ in time $n^{o(k)}$, then the ETH is false. \end{theorem}

\begin{proof} The described construction, as well as Lemmas \ref{lem: forward
dir} and \ref{lem: backward dir} give a reduction from \msc, which is W[2]-hard
parameterized by $k$, to an instance of \free\ with $k+1$ colors, where the
question is to decide if $\optfree(G,\col)\le 2k$.  Furthermore, it is known
that \msc\ generalizes \textsc{Dominating Set}, which does not admit an
algorithm running in time $n^{o(k)}$, under the ETH \cite{CyganFKLMPPS15}.
Since our reduction only modifies $k$ by a constant, we odtain the same result
for \free.  \end{proof}

We note that because of Lemma \ref{lem:col-1} we can always assume that the
number of colors of a given instance is not much higher than the length of the
optimal solution. As a result, \free\ parameterized by $\optfree$ is equivalent
to the parameterization of \free\ by $\optfree+c_{\max}$ and the result of
Theorem \ref{thm:whard} also applies to this parameterization.

\subsection{Kernel lower bound for {\free}}
As a byproduct of the reduction above, we can show a kernel lower bound for {\free}
parameterized by the vertex cover number.

Let $P$ and $Q$ be parameterized problems.
A polynomial-time computable function $f \colon \Sigma^{*} \times N \to \Sigma^{*} \times N$
is a \emph{polynomial parameter transformation} from $P$ to $Q$
if there is a polynomial $p$ such that for all $(x,k) \in \Sigma^{*} \times N$,
\begin{itemize}
  \item $(x,k) \in P$ if and only if $(x',k') = f(x,k) \in Q$, and
  \item $k' \le p(k)$.
\end{itemize}
If such a function exits, then $P$ is \emph{polynomial parameter reducible} to $Q$.

\begin{proposition}
[\cite{BodlaenderDFH09}]
\label{prop:ppt-poly-kernel}
Let $P$ and $Q$ be parameterized problems,
and $P'$ and $Q'$ be unparameterized versions of $P$ and $Q$, respectively.
Suppose $P'$ is NP-hard, $Q'$ is in NP,
and $P$ is polynomial parameter reducible to $Q$.
If $Q$ has a polynomial kernel, then $P$ also has a polynomial kernel.
\end{proposition}

\begin{theorem}
\label{thm:no-poly-kernel/vc}
{\free} parameterized by the vertex cover number 
admits no polynomial kernel unless $\mathrm{PH} = \Sigma^{\mathrm{p}}_{3}$.
\end{theorem}
\begin{proof}
The reduction in this section can be seen as a polynomial parameter transformation
from {\msc} parameterized by the solution size $k$ and the size $|R|$ of the universe
to {\free} parameterized by the vertex cover number
with a polynomial $p(k, |R|) = 3k^{2} + |R|$.
To see this observe that the black vertices in \figref{Fig: free-flood-it} form a vertex cover of size $3k^{2} + |R|$.

Since {\msc} is NP-hard and the decision version of {\free} is in NP,
Proposition~\ref{prop:ppt-poly-kernel} implies that
if {\free} parameterized by the vertex cover number has a polynomial kernel,
then {\msc} parameterized by $k$ and $|R|$ also has a polynomial kernel.

It is known that {\setc} (and thus {\msc}) parameterized simultaneously by $k$ and $|R|$
does not admit a polynomial kernel unless $\mathrm{PH} = \Sigma^{\mathrm{p}}_{3}$~\cite{DomLS14}.
This completes the proof.
\end{proof}


\section{Clique-width and the number of moves}
\label{sec:clique-width-MSO}
In this section, we consider as a combined parameter for {\free} the length of an optimal solution and the clique-width.
We show that this case is indeed fixed-parameter tractable by using the theory of the monadic second-order logic on graphs.
As an application of this result,
we also show that combined parameterization by the number of colors and the modular-width is fixed-parameter tractable.

To prove the main claim, 
we show that {\free} with a constant length of optimal solutions is an MSO$_{1}$-definable decision problem.
The syntax of MSO$_{1}$ (one-sorted monadic second-order logic) of graphs includes
(i) the logical connectives $\lor$, $\land$, $\lnot$, $\Leftrightarrow$, $\Rightarrow$,
(ii) variables for vertices and vertex sets,
(iii) the quantifiers $\forall$ and $\exists$ applicable to these variables, and
(iv) the following binary relations:
\begin{itemize}
  \item $u \in U$ for a vertex variable $u$ and a vertex set variable $U$;
  \item $\mathbf{adj}(u,v)$ for two vertex variables $u$ and $v$,
  where the interpretation is that $u$ and $v$ are adjacent;
  \item equality of variables.
\end{itemize}
If $G$ models an MSO$_{1}$ formula $\varphi$ with an assignment $X_{1}, \dots, X_{q} \subseteq V(G)$ to the $q$ free variables in $\varphi$,
then we write $\langle G, X_{1}, \dots, X_{q} \rangle \models \varphi$.

It is known that, given a graph of clique-width at most $w$, an MSO$_{1}$ formula $\varphi$,
and an assignment to the free variables in $\varphi$,
the problem of deciding whether $G$ models $\varphi$ with the given assignment
is solvable in time $O(f(\lvert\lvert\varphi\rvert\rvert, w) \cdot n^{3})$,
where $f$ is a computable function
and $\lvert\lvert\varphi\rvert\rvert$ is the length of $\varphi$~\cite{CourcelleMR00,Oum08}.

\begin{theorem}
\label{thm:cwd+steps}
Given an instance $(G, \col)$ of {\free}
such that $G$ has $n$ vertices and clique-width at most $w$,
it can be decided in time $O(f(k, w) \cdot n^{3})$
whether $\optfree(G, \col) \le k$, where $f$ is some computable function.
\end{theorem}
\begin{proof}
Let $V_{i}$ ($1 \le i \le c_{\max}$) be the set of color $i$ vertices in the input graph.
We construct an MSO$_{1}$ formula $\varphi$ with $c_{\max}$ free variables $X_{1}, \dots, X_{c_{\max}}$
such that $\optfree(G, \col) \le k$ if and only if 
$G$ models $\varphi$ with the assignment $X_{i} := V_{i}$ for $1 \le i \le c_{\max}$.
We can define the desired formula $\varphi(X_{1}, \dots, X_{c_{\max}})$ as follows:
\begin{align*}
  \varphi(X_{1}, \dots, X_{c_{\max}}) := 
  \textstyle\bigvee_{1 \le c_{1}, \dots, c_{k} \le k}
  \exists_{v_{1}, v_{2}, \dots, v_{k} \in V(G)} 
  \textstyle\bigvee_{1 \le c \le k}
  \forall_{u \in V(G)} \ \mathbf{color}_{c,k}(u),
\end{align*}
where $\mathbf{color}_{c, i}(u)$ for $0 \le i \le k$
implies that the color of $u$ is $c$ after the moves $(v_{1}, c_{1}), \dots (v_{i}, c_{i})$.

We define $\mathbf{color}_{c, i}(u)$ recursively as follows.
We first set $\mathbf{color}_{c, 0}(u) := (u \in X_{i})$.
This is correct as we assign $V_{i}$ to $X_{i}$.
For $1 \le i \le k$, we set
\begin{align*}
  \mathbf{color}_{c, i}(u)
  =
  \begin{cases}
    \mathbf{color}_{c, i-1}(u) \lor \mathbf{SameCCC}_{i-1}(u, v_{i}) & c = c_{i},
    \\
    \mathbf{color}_{c, i-1}(u) \land \lnot \mathbf{SameCCC}_{i-1}(u, v_{i}) & c \ne c_{i},
  \end{cases}
\end{align*}
where $\mathbf{SameCCC}_{i}(u_{1}, u_{2})$ implies that
$u_{1}$ and $u_{2}$ are in the same monochromatic component after the moves $(v_{1}, c_{1}), \dots (v_{i}, c_{i})$.
The formula precisely represent the recursive nature of the color of the vertices.
That is, a vertex $u$ is of color $c$ after the $i$th move $(v_{i}, c_{i})$ if and only if
either 
its color is changed to $c$ by the $i$ move,
or it was already of color $c$ before the $i$th move and its color is not changed by the $i$ move.

Given that $\mathbf{color}_{c, i}(u)$ is defined for all $c$ and $u$,
defining $\mathbf{SameCCC}_{i}(u_{1}, u_{2})$ is a routine:
\begin{align*}
  \mathbf{SameCCC}_{i}(u_{1}, u_{2})
  =
  \textstyle\bigvee_{1 \le c \le k} \exists_{S \subseteq V(G)} 
  &(u_{1}, u_{2} \in S) \land (\forall_{u \in S}\ \mathbf{color}_{c, i}(u))
  \\
  &\land (\forall_{T \subseteq S}\ (T = \emptyset) \lor \exists_{x \in T} \exists_{y \in S \setminus T}\ \mathbf{adj}(x,y)).
\end{align*}

Since $k \ge c_{\max} -1$ by Lemma~\ref{lem:col-1},
it holds that $\lvert\lvert\varphi\rvert\rvert$ is bounded by a function of $k$.
\end{proof}

\begin{corollary}
\label{cor:mwd+steps}
Given an integer $k$ and an instance $(G, \col)$ of {\free}
such that $G$ has $n$ vertices and modular-width at most $w$,
it can be decided in time $O(f(c_{\max}, w) \cdot n^{3})$
whether $\optfree(G, \col) \le k$, where $f$ is some computable function.
\end{corollary}
\begin{proof}
Observe that for every connected graph $G$ of modular-width at most $w$,
it holds that $\optfree(G, \col) \le w + c_{\max} - 2$:
we pick one vertex $v$ from a module $M$;
we next color one vertex in each module except $M$ with $\col(v)$;
we then play at $v$ with the remaining $c_{\max} - 1$ colors.
Thus we can assume that $k \le w + c_{\max} - 2$.
Since the modular-width of a graph is at most its clique-width by their definitions,
Theorem~\ref{thm:cwd+steps} gives an $O(g(w+c_{\max}, w) \cdot n^{3})$-time algorithm for some computable $g$,
which can be seen as an $O(f(c_{\max}, w) \cdot n^{3})$-time algorithm for some computable $f$.
\end{proof}


\section{Neighborhood diversity and the number of colors}
\label{sec:nd-kernel}
Since the modular-width of a graph is upper bounded by its neighborhood diversity,
Corollary~\ref{cor:mwd+steps} in the previous section implies that
{\free} is fixed-parameter tractable when parameterized by both the neighborhood diversity and the number of colors.
Here we show that {\free} admits a polynomial kernel with the same parameterization.
This section is devoted to a proof of the following theorem.
\begin{theorem}
\label{thm:nd+cmax_poly-kernel}
{\free} admits a kernel of at most $\nd(G) \cdot c_{\max} \cdot (\nd(G) + c_{\max} - 1)$ vertices.
\end{theorem}

Fellows et al.~\cite{FellowsPRSS17} observed that for {\fixed}, 
a polynomial kernel with the same parameterization can be easily obtained
since twin vertices of the same color can be safely contracted.
In {\free}, this is true for true twins but not for false twins.
See \figref{fig:false-twin}.
\begin{figure}[tb]
  \centering
  \begin{minipage}[b]{0.45\linewidth}
    \centering
    \includegraphics[scale=0.7]{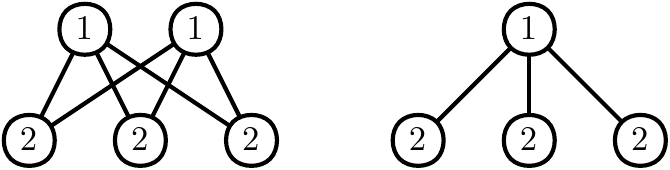}
    \subcaption{Removing a false twin is not safe.}\label{fig:false-twin}
  \end{minipage}
  \begin{minipage}[b]{0.45\linewidth}
    \centering
    \includegraphics[scale=0.7]{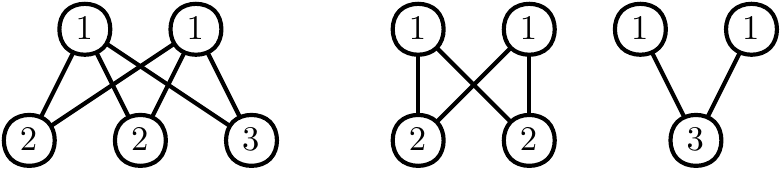}
    \subcaption{Removing a twin color is not safe.}\label{fig:twin-colors}
  \end{minipage}
  \caption{Simple reductions that do not work for {\free}.}
\end{figure}

Though it might be still possible to show something like 
``if there are more than some constant number of false twins with the same color,
then one can remove one of them without changing the minimum number of moves,''
here we show a weaker claim. Our reduction rules are as follows:
\begin{itemize}
  \item \noindent Rule \textit{TT}:
  Let $u$ and $v$ be true twins of the same color in $(G,\col)$.
  Remove $v$.

  \item \noindent Rule \textit{FT}:
  Let $F$ be a set of false-twin vertices of the same color in $(G,\col)$
  such that $|F| = \nd(G)  + c_{\max}$.
  Remove arbitrary one vertex in $F$.
\end{itemize}

Observe that after applying TT and FT exhaustively in polynomial time,
the obtained graph can have at most $\nd(G) \cdot c_{\max} \cdot (\nd(G) + c_{\max} - 1)$ vertices.
This is because each set of twin vertices can contain at most $\nd(G) + c_{\max}-1$ vertices.
Hence, to prove Theorem~\ref{thm:nd+cmax_poly-kernel}, it suffices to show the safeness of the rules.

\begin{lemma}
Rule TT is safe.
\end{lemma}
\begin{proof}
  Let $u$ and $v$ be true twins of the same color.
  Observe that removing $v$ is equivalent to contracting the edge $\{u,v\}$.
  Since $u$ and $v$ are in the same monochromatic component,
  the lemma holds.
\end{proof}

To guarantees the safeness of FT, we need the following technical lemmas.
\begin{lemma}
\label{lem:unplayed-false-twins}
Let $(G,\col)$ be an instance of {\free}
and $x, y \in V(G)$ be false-twin vertices of the same color $c$.
A sequence $(u_{1}, c_{1}), \dots, (u_{k}, c_{k})$ with $u_{i} \notin \{x,y\}$ for $1 \le i \le k$
is a valid flooding sequence for $(G,\col)$ if and only if it is a valid flooding sequence for $(G-x,\col|_{G-x})$.
\end{lemma}
\begin{proof}
Let $(G',\col') = (G-x,\col|_{G-x})$.
If a neighbor of $x$ and $y$ has color $c$, then the lemma trivially holds.
Hence, in what follows, we assume that none of the vertices adjacent to $x$ and $y$ has color $c$.
Assume that $(u_{1}, c_{1}), \dots, (u_{k}, c_{k})$ is valid for at least one of $(G,\col)$ and $(G',\col')$.
Then there is a move that changes the color of a neighbor of $y$ to $c$ since $u_{i} \ne y$ for $1 \le i \le k$.
Let $(u_{i}, c_{i})$ be the first such move.

The first part $(u_{1}, c_{1}), \dots, (u_{i-1}, c_{i-1})$ of the sequence 
has the same effect to $(G,\col)$ and $(G',\col')$.
That is, the monochromatic components and connection among them
are the same in $(G,\col_{i-1})$ and $(G',\col'_{i-1})$ except that $(G,\col_{i-1})$ contains the monochromatic component $\{x\}$.
Note that $\{y\}$ is a monochromatic component of color $c$ in both $(G,\col_{i-1})$ and $(G',\col'_{i-1})$,
and $\{y\}$ and $\{x\}$ have the same adjacent monochromatic components in $(G,\col_{i-1})$.

Let $\mathcal{C}$ be the set of monochromatic components of color $c$ in $(G, \col_{i-1})$ that
are adjacent to $\comp(\col_{i-1}, u_{i})$.
Similarly, let $\mathcal{C}'$ be the set of monochromatic components of color $c$ in $(G', \col'_{i-1})$ that
are adjacent to $\comp(\col'_{i-1}, u_{i})$.
Observe that $\comp(\col_{i-1}, u_{i}) = \comp(\col'_{i-1}, u_{i})$,
$\mathcal{C}' = \mathcal{C} \setminus \{\{x\}\}$, and
$\{y\} \in \mathcal{C} \cap \mathcal{C}'$.

Now we apply the move $(u_{i}, c_{i})$.
It follows that $\comp(\col_{i}, u_{i}) = \comp(\col_{i-1}, u_{i}) \cup \bigcup_{C \in \mathcal{C}} C$
and $\comp(\col'_{i}, u_{i}) = \comp(\col'_{i-1}, u_{i}) \cup \bigcup_{C' \in \mathcal{C'}} C'$.
Since $\comp(\col_{i}, u_{i}) \setminus \comp(\col'_{i}, u_{i}) = \{x\}$
and both $\comp(\col_{i}, u_{i})$ and $\comp(\col'_{i}, u_{i})$ include $y$,
the components $\comp(\col_{i}, u_{i})$ and $\comp(\col'_{i}, u_{i})$
have the same adjacent monochromatic components.
Also, for each $u \notin \comp(\col_{i}, u_{i})$,
$\comp(\col_{i-1}, u) = \comp(\col'_{i-1}, u)$ holds.

This implies that $(G, \col_{i})$ and $(G', \col'_{i})$ are equivalent
and thus the remaining of the sequence, i.e.\ $(u_{i+1}, c_{i+1}), \dots, (u_{k}, c_{k})$, 
has the same effect to $(G',\col'_{i})$ and $(G,\col_{i})$.
Therefore, $(G,\col_{k})$ is constant if and only if so is $(G',\col'_{k})$.
\end{proof}

\begin{lemma}
Let $S$ be a set of false-twin vertices with the same color in $(G, \col)$.
If $|S| \ge \optfree(G,\col) + 2$, then $\optfree(G,\col) = \optfree(G-x,\col|_{G-x})$ for every $x \in S$.
\end{lemma}
\begin{proof}
\label{lem:huge-colored-false-twin-class}
We first show that $\optfree(G,\col) \le \optfree(G-x,\col|_{G-x})$ for every $x \in S$.
Let $(u_{1}, c_{1}), \dots, (u_{k}, c_{k})$ be an optimal valid flooding sequence for $(G - x, \col|_{G-x})$.
Assume that $\optfree(G,\col) \ge k$ (otherwise we are done).
Since $|S \setminus \{x\}| \ge \optfree(G,\col) + 1 \ge k+1$,
there is a vertex $y \in S \setminus \{x\}$ such that $u_{i} \ne y$ for $1 \le i \le k$.
By Lemma~\ref{lem:unplayed-false-twins}, $(u_{1}, c_{1}), \dots, (u_{k}, c_{k})$ is valid for $(G,\col)$ as well.

Next we show the other direction.
Since $S$ is a set of monochromatic false-twin vertices,
it suffices to show that $\optfree(G,\col) \ge \optfree(G-x,\col|_{G-x})$ for some $x \in S$.
Let $(u_{1},c_{1}), \dots, (u_{k},c_{k})$ be an optimal valid flooding sequence of $(G, \col)$.
Since $|S| \ge \optfree(G,\col) + 2 = k + 2$,
there are two vertices $x, y \in S$ such that $u_{i} \notin \{x,y\}$ for $1 \le i \le k$.
By Lemma~\ref{lem:unplayed-false-twins}, $(u_{1}, c_{1}), \dots, (u_{k}, c_{k})$ is valid for $(G-x,\col|_{G-x})$ as well.
\end{proof}

\begin{corollary}
Rule FT is safe.
\end{corollary}
\begin{proof}
Let $G$ be a connected graph
and $\col$ a coloring of $G$ with $c_{\max}$ colors.
Observe that $(G,\col)$ admits a flooding sequence of length $\nd(G) + c_{\max} - 2$ as follows.
Let $T$ be a maximal set of twin vertices of $G$ and $c$ be a color used in $T$.
For each maximal set of twin vertices $T' \ne T$ of $G$, pick a vertex $u \in T'$ and play the move $(u,c)$.
After these $\nd(G)-1$ moves, the vertices of color $c$ form a connected dominating set of $G$.
Now pick a vertex $v$ of color $c$ and play the move $(v,c')$ for each $c' \in [c_{\max}] \setminus \{c\}$.
These $c_{\max} -1$ moves make the coloring constant.

Now for some color class $C$ and a false-twin class $I$ of $(G,\col)$,
if $|C \cap I| \ge c_{\max} + \nd(G)$, then we can remove an arbitrary vertex in $C \cap I$
while preserving the optimal number of steps by Lemma~\ref{lem:huge-colored-false-twin-class}.
This implies the safeness of FT.
\end{proof}

Note that using the concept of twin colors,
Fellows et al.~\cite{FellowsPRSS17} further reduced the number of colors in instances of {\fixed}
and obtained a (nonpolynomial-size) kernel parameterized by the neighborhood diversity.
They say that two colors are \emph{twin} if the colors appear in the same family of the maximal sets of twin vertices.
They observed that, in {\fixed}, removing one of twin colors reduced the fewest number of moves exactly by 1.
Unfortunately, this is not the case for {\free}. See~\figref{fig:twin-colors}.


\section{Relation Between Fixed and Free Flood-It}
\label{sec:fixed<=2free}
The main theorem of this section is the following:

\begin{theorem}\label{thm:ineq}

For any graph $G=(V,E)$, coloring function $\col$ on $G$, and $p\in V$ we have
\[
\optfree(G,\col) \le \optfixed(G,\col,p) \le 2\optfree(G,\col).  
\]

\end{theorem}

Theorem \ref{thm:ineq} states that the optimal solutions for \free\ and \fixed\
can never be more than a factor of $2$ apart. It is worthy of note that we
could not hope to obtain a constant smaller than $2$ in such a theorem, and
hence the theorem is tight.

\begin{theorem} There exist instances of \fixed\ such that $\optfixed(G,\col,p)
= 2\optfree(G,\col)$ \end{theorem}

\begin{proof} Consider a path on $2n+1$ vertices properly colored with colors
$1,2$. If we set the pivot to be one of the endpoints then $\optfree=2n$.
However, it is not hard to obtain a \free\ solution with $n$ moves by playing
every vertex at odd distance from the pivot.  \end{proof}

Before we proceed to give the proof of Theorem \ref{thm:ineq}, let us give a
high-level description of our proof strategy and some general intuition. The
first inequality is of course trivial, so we focus on the second part. We will
establish it by induction on the number of non-pivot moves performed by an
optimal \free\ solution. The main inductive argument is based on observing that
a valid \free\ solution will either at some point play a neighbor $u$ of the
component of $p$ to give it the same color as $p$, or if not, it will at some
point play $p$ to give it the same color as one of its neighbors.  The latter
case is intuitively easier to handle, since then we argue that the move that
changed $p$'s color can be performed first, and if the first move is a pivot
move we can easily fall back on the inductive hypothesis.  The former case,
which is the more interesting one, can be handled by replacing the single move
that gives $u$ the same color as $p$, with two moves: one that gives $p$ the
same color as $u$, and one that flips $p$ back to its previous color.
Intuitively, this basic step is the reason we obtain a factor of $2$ in the
relationship between the two versions of the game.

The inductive strategy described above faces some complications due to the fact
that rearranging moves in this way may unintentionally re-color some vertices,
which makes it harder to continue the rest of the solution as before. To avoid
this we define a somewhat generalized version of \free, called \sfree.

\begin{definition}

Given $G=(V,E)$, a coloring function $\col$ on $G$, and a pivot $p\in V$, a
\emph{set-move} is a pair $(S,c)$, with $S\subseteq V$ and $S=\comp(\col,u)$
for some $u\in V$, or $\{p\}\subseteq S\subseteq \comp(\col,p)$.  The result of
$(S,c)$ is the coloring $\col'$ that sets $\col'(v)=c$ for $v\in S$; and
$\col'(v)=\col(v)$ otherwise.

\end{definition}

We define \sfree\ as the problem of determining the minimum number of set-moves
required to make a graph monochromatic, and \sfixed\ as the same problem when
we impose the restriction that every move must change the color of $p$, and
denote as $\optsfree, \optsfixed$ the corresponding optimum values.

Informally, a set-move is the same as a normal move in \free, except that we
are also allowed to select an arbitrary connected monochromatic set $S$ that
contains $p$ (even if $S$ is not maximal) and change its color. Intuitively,
one would expect moves that set $S$ to be a proper subset of $\comp(\col,p)$ to
be counter-productive, since such moves split a monochromatic component into
two pieces. Indeed, we prove below in Lemma \ref{lem:subset} that the optimal
solutions to \fixed\ and \sfixed\ coincide, and hence such moves do not help.
The reason we define this version of the game is that it gives us more freedom
to define a solution that avoids unintentionally recoloring vertices as we
transform a given \free\ solution to a \fixed\ solution.

\begin{lemma}\label{lem:subset}

For any graph $G=(V,E)$, coloring function $\col$ on $G$, and pivot $p\in V$ we
have $\optfixed(G,\col,p) = \optsfixed(G,\col,p)$.

\end{lemma}

\begin{proof}

First, observe that $\optsfixed(G,\col,p)\le \optfixed(G,\col,p)$ is trivial,
as any solution of \fixed\ is a solution to \sfixed\ by playing the same
sequence of colors and always selecting all of the connected monochromatic
component of $p$.  

Let us also establish the converse inequality.  Consider a solution $(S_1,c_1),
(S_2,c_2), \ldots, (S_k,c_k)$ of \sfixed, where by definition we have $p\in
S_i$ for all $i\in [k]$. We would like to prove that $(p,c_1), (p,c_2), \ldots,
(p,c_k)$ is a valid solution for \fixed. Let $\col_i$ be the result of the
first $i$ set-moves of the former solution, and $\col'_i$ be the result of the
first $i$ moves of the latter solution. We will establish by induction the
following:

\begin{enumerate}

\item For all $i\in [k]$ we have $\comp(\col_i,p)\subseteq \comp(\col_i',p)$.

\item For all $i\in [k], u\in V\setminus \comp(\col'_i,p)$ we have
$\col_i(u)=\col'_i(u)$.

\end{enumerate}

The statements are true for $i=0$. Suppose that the two statements are true
after $i-1$ moves.  The first solution now performs the set-move $(S_i,c_i)$
with $S_i\subseteq \comp(\col_{i-1},p) \subseteq \comp(\col'_{i-1},p)$. We now
have that $\comp(\col_i,p)$ contains $S_i$ plus the neighbors of $S_i$ which
have color $c_i$ in $\col_{i-1}$. Such vertices either also have color $c_i$ in
$\col'_{i-1}$, or are contained in $\comp(\col'_{i-1},p)$; in both cases they
are included in $\comp(\col'_i,p)$, which establishes the first condition. To
see that the second condition continues to hold observe that every vertex for
which $\col_{i-1}(u)\neq \col_i(u)$ or $\col'_{i-1}(u)\neq\col'_i(u)$ belongs
in $\comp(\col'_i,p)$; the colors of other vertices remain unchanged.  Since in
the end $\comp(\col_k,p)=V$ the first condition ensures that
$\comp(\col'_k,p)=V$.  \end{proof}

We are now ready to state the proof of Theorem \ref{thm:ineq}.

\begin{proof}[Proof of Theorem \ref{thm:ineq}]

As mentioned, we focus on proving the second inequality as the first inequality
follows trivially from the definition of the problems. Given a graph $G=(V,E)$,
an initial coloring function $\col=\col_0$, and a pivot $p\in V$, we suppose we
have a solution to \free\ $(u_1,c_1), (u_2,c_2), \ldots, (u_k,c_k)$. In the
remainder, we denote by $\col_i$ the coloring that results after the moves
$(u_1,c_1),\ldots,(u_i,c_i)$.   We can immediately construct an equivalent
solution to \sfree\ from this, producing the same sequence of colorings:
$(\comp(\col_0,u_1),c_1), (\comp(\col_1,u_2),c_2),\ldots,
(\comp(\col_{k-1},u_k),c_k)$. We will transform this solution to a solution of
$\sfixed$ of length at most $2k$, and then invoke Lemma \ref{lem:subset} to
obtain a solution for \fixed\ of length at most $2k$.  More precisely, we will
show that for any $G,\col, p$ we have $\optsfixed(G,\col,p) \le 2
\optsfree(G,\col,p)$.

For a solution $\mathcal{S}= (S_1,c_1), (S_2,c_2),\ldots, (S_k,c_k)$ to \sfree\
we define the number of bad moves of $\mathcal{S}$ as $b(\mathcal{S})=|\{
(S_i,c_i)\ |\ p\not\in S_i\}|$. We will somewhat more strongly prove the
following statement for all $G,\col,p$: for any valid \sfree\ solution
$\mathcal{S}$, we have $$\optsfixed(G,\col,p) \le |\mathcal{S}| +
b(\mathcal{S})$$ 

Since $|\mathcal{S}| + b(\mathcal{S}) \le 2|\mathcal{S}|$, the above statement
will imply the promised inequality and the theorem. 

We prove the statement by induction on $|\mathcal{S}|+2b(\mathcal{S})$. If
$|\mathcal{S}|+2b(\mathcal{S}) \le 2$ then $\mathcal{S}$ is already a \sfixed\
solution, so the statement is trivial. Suppose then that the statement holds
when $|\mathcal{S}|+2b(\mathcal{S})\le n$ and we have a solution $\mathcal{S}$
with $|\mathcal{S}|+2b(\mathcal{S})=n+1$. We consider the following cases:

$\bullet$ The first move $(S_1,c_1)$ has $p\in S_1$. By the inductive
hypothesis there is a \sfixed\ solution of length at most
$|\mathcal{S}|+b(\mathcal{S})-1$ for $(G,\col_1,p)$. We build a solution for
\sfixed\ by appending this solution to the move $(S_1,c_1)$, since this is a
valid move for \sfixed.

$\bullet$ There exists a move $(S_i,c_i)$ with $S_i=\comp(\col_{i-1},u)$, for
some $u\in N(\comp(\col_{i-1},p))\setminus \comp(\col_{i-1},p)$ such that
$c_i=\col_{i-1}(p)$. That is, there exists a move that plays a vertex $u$ that
currently has a different color than $p$, and as a result of this move the
component of $u$ and $p$ merge, because $u$ receives the same color as $p$ and
$u$ has a neighbor in the component of $p$.  

Consider the first such move.  We build a solution $\mathcal{S'}$ as follows:
we keep moves $(S_1,c_1)\ldots (S_{i-1},c_{i-1})$; we add the moves
$(\comp(\col_{i-1},p),\col_{i-1}(u)), (\comp(\col_{i-1},p)\cup
\comp(\col_{i-1},u), \col_{i-1}(p))$; we append the rest of the previous
solution $(S_{i+1},c_{i+1}),\ldots$.

To see that $\mathcal{S'}$ is still a valid solution we observe that
$\comp(\col_{i-1},p)\cup \comp(\col_{i-1},u)$ is monochromatic and connected
when we play it, and that the result of the first $i-1$ moves, plus the two new
moves is exactly $\col_i$. We also note that
$\mathcal{S'}+b(\mathcal{S'})=\mathcal{S}+b(\mathcal{S})$ because we replaced
one bad move with two good moves. However,
$\mathcal{S'}+2b(\mathcal{S'})<\mathcal{S}+2b(\mathcal{S})$, hence by the
inductive hypothesis there exists a \sfixed\ solution of the desired length.

$\bullet$ There does not exist a move as specified in the previous case. We
then show that this reduces to the first case. If no move as described in the
previous case exists and the initial coloring is not already constant,
$\mathcal{S}$ must have a move $(S_i,c_i)$ where $\{p\}\subseteq S_i\subseteq
\comp(\col_0,p)$ and $c_i = \col_{i-1}(u)$ for $u\in
N(\comp(\col_0,p))\setminus \comp(\col_0,p)$. In other words, this is a good
move (it changes the color of $p$), that adds a new vertex $u$ to the connected
monochromatic component of $p$. Such a move must exist, since if the initial
coloring is not constant, the initial component of $p$ must be extended, and we
assumed that no move that extends it by recoloring one of its neighbors exists.

Consider the first such good move  $(S_i,c_i)$ as described above. We build a
solution $\mathcal{S'}$ as follows: the first move is
$(\comp(\col_0,p),\col_0(u))$, where $u$ is, as described above, the neighbor
of $\comp(\col_0,p)$ with $\col_{i-1}(u)=c_i$.  For $j\in [i-1]$ we add the
move $(S_j,c_j)$ if $u\not\in S_j$, or the move
$(\comp(\col_{j-1},u)\cup\comp(\col_0,p),c_j)$ if $u\in S_j$. In other words,
we keep other moves unchanged if they do not affect $u$, otherwise we add to
them $\comp(\col_0,p)$. We observe that these moves are valid since we maintain
the invariant that $\comp(\col_0,p)$ and $u$ have the same color and since none
of the first $i-1$ moves of $\mathcal{S}$ changes the color of $p$ (since we
selected the first such move).  The result of these $i$ moves is exactly
$\col_i$. We now append the remaining move $(S_{i+1},c_{i+1}),\ldots$, and we
have a solution that starts with a good move, has the same length and the same
(or smaller) number of bad moves as $\mathcal{S}$ and is still valid.  We have
therefore reduced this to the first case.  \end{proof}

As we mentioned before, this combinatorial theorem implies $2$-approximability of {\free}
for the cases where {\fixed} is polynomial-time solvable.
Also, as {\fixed} admits a $(c_{\max}-1)$ approximation~\cite{CliffordJMS12},\footnote{%
Their proof was only for grids, but it just works for the general case.}
we have a $2(c_{\max}-1)$ approximation for {\free}.
\begin{corollary}
{\free} admits a $2(c_{\max}-1)$ approximation, where $c_{\max}$ is the number of used colors.
\end{corollary}


\section{Non-monotonicity of {\free}}
\label{sec:non-monotonicity}
We now consider the (non-)monotonicity of the problem.
A game has the \emph{monotonicity property} if no legal move makes the situation worse.
That is, if {\fixed} (or {\free}) has the monotonicity property,
then no single move increases the minimum number of steps to make the input graph monotone.
We believe that the monotonicity of {\fixed} was known as folklore
and used implicitly in the literature.
On the other hand, we are not sure that 
the non-monotonicity of {\free} was widely known.
The only result we are aware of is by Meeks and Scott~\cite{MeeksS12}
who showed that on paths {\free} has the monotonicity property.
In the following, 
we show that {\free} loses its monotonicity property
as soon as the underlying graph becomes a path with one attached vertex.

To be self-contained, we start with proving the following folklore,
which says that {\fixed} is monotone.
\begin{lemma}
Let $(G,\col)$ is an instance of {\fixed} with pivot $p$.
For every color $c$, it holds that $\optfixed(G,\col') \le \optfixed(G,\col)$,
where $(G,\col')$ is the result of the move $(p,c)$.
\end{lemma}
\begin{proof}
Let $(p, c_{1}), \dots, (p, c_{k})$ be an optimal solution for $(G,\col)$.
We show that this sequence is valid for $(G,\col')$ too.
Let $\col = \col_{0}$ and for $i \ge 1$, let $\col_{i}$ be the coloring obtained from $\col_{i-1}$ by applying the $i$th move $(p,c_{i})$.
We define $\col'_{i}$ in the analogous way.
Observe that since all moves are played on the pivot $p$,
we have $\comp(\col_{i},v) \in \{\comp(\col,v), \comp(\col_{i},p)\}$
and $\comp(\col'_{i},v) \in \{\comp(\col',v), \comp(\col'_{i},p)\}$
for every $v \in V(G)$ and for every $i$.

It suffices to show that $\comp(\col_{i}, p) \subseteq \comp(\col'_{i}, p)$ for all $i$.
This obviously holds when $i = 0$.
Assume that $\comp(\col_{i}, p) \subseteq \comp(\col'_{i}, p)$ for some $i \ge 0$.
Let $v \in V(G) \setminus \{p\}$ such that $\comp(\col,v)$ is not contained in but adjacent to $\comp(\col_{i},p)$,
and thus contained in $\comp(\col_{i+1},p)$.
Since $\comp(\col_{i}, p) \subseteq \comp(\col'_{i}, p)$, 
the monochromatic component $\comp(\col,v)$ is either contained in or adjacent to $\comp(\col_{i}',p)$.
Therefore, $\comp(\col,v)$ is contained in $\comp(\col_{i+1}',p)$.
This implies that $\comp(\col_{i+1}, p) \subseteq \comp(\col_{i+1}', p)$.
\end{proof}

Meeks and Scott~\cite{MeeksS12} showed the following monotonicity of {\free} on paths.
\begin{proposition}
[{\cite[Lemma~4.1]{MeeksS12}}]
Let $G$ be a path, $\col$ a vertex coloring of $G$,
$(v, c)$ a move, and $(G,\col')$ the result of the move.
Then, $\optfree(G,\col') \le \optfree(G,\col)$.
\end{proposition}
One may wonder whether the monotonicity property holds in general for {\free}.
The following example, however, shows that it does not hold even for some graphs very close to paths.
See the two instances in \figref{fig:none-monotone-min}.
The instance $(G,\col')$ is obtained from $(G,\col)$ by playing the move $(v,3)$.
We show that $\optfree(G,\col) < \optfree(G,\col')$.

Observe that $\optfree(G,\col) = 3$: by Lemma~\ref{lem:col}, $\optfree(G,\col) \ge 3$,
and the sequence $(u,2)$, $(u,1)$, $(u,1)$ floods the graph.
Suppose that $\optfree(G,\col') \le 3$.
By Lemma~\ref{lem:col}, the first move in each optimal solution of $(G,\col')$
has to make the subgraph induced by some color connected,
and then the second move has to remove the connected color.
We can see that 2 is the only color that can play this role.
If the first move is not played on $u$,
then it is played on one of the two color-2 vertices
and then the second move is played on the other color-2 vertex.
Such a sequence of two moves cannot make either of color-1 or color-3 vertices connected.
Thus by Lemma~\ref{lem:col}, it still needs at least two moves.
Hence we can conclude that the first move is $(u,2)$.
Now the second move has to remove the color 2,
and thus has to be played on $u$ (or equivalently on any vertex in the monochromatic component including $u$).
No matter which color we choose, we end up with an instance with at least two colors that are not connected.
Again by Lemma~\ref{lem:col}, this instance needs more than one step, and thus $\optfree(G,\col') > 3$.
\begin{figure}[tb]
  \centering
  \includegraphics[scale=0.7]{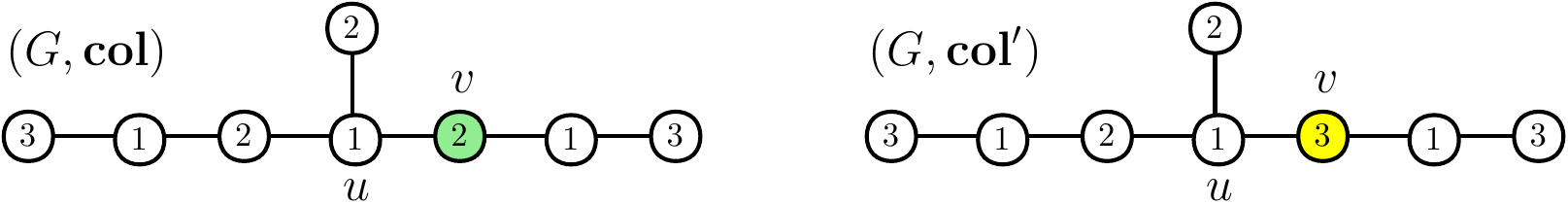}
  \caption{Non-monotonicity of {\free}.}
  \label{fig:none-monotone-min}
\end{figure}







\bibliographystyle{plainurl}
\bibliography{ffip}

\providecommand{\noopsort}[1]{}
\begin{thebibliography}{10}

\bibitem{BodlaenderDFH09}
Hans~L. Bodlaender, Rodney~G. Downey, Michael~R. Fellows, and Danny Hermelin.
\newblock On problems without polynomial kernels.
\newblock {\em J. Comput. Syst. Sci.}, 75(8):423--434, 2009.
\newblock \href {http://dx.doi.org/10.1016/j.jcss.2009.04.001}
  {\path{doi:10.1016/j.jcss.2009.04.001}}.

\bibitem{CliffordJMS12}
Rapha{\"{e}}l Clifford, Markus Jalsenius, Ashley Montanaro, and Benjamin Sach.
\newblock The complexity of flood filling games.
\newblock {\em Theory Comput. Syst.}, 50(1):72--92, 2012.
\newblock \href {http://dx.doi.org/10.1007/s00224-011-9339-2}
  {\path{doi:10.1007/s00224-011-9339-2}}.

\bibitem{CourcelleMR00}
Bruno Courcelle, Johann~A. Makowsky, and Udi Rotics.
\newblock Linear time solvable optimization problems on graphs of bounded
  clique-width.
\newblock {\em Theory Comput. Syst.}, 33(2):125--150, 2000.
\newblock \href {http://dx.doi.org/10.1007/s002249910009}
  {\path{doi:10.1007/s002249910009}}.

\bibitem{CyganFKLMPPS15}
Marek Cygan, Fedor~V. Fomin, Lukasz Kowalik, Daniel Lokshtanov, D{\'{a}}niel
  Marx, Marcin Pilipczuk, Michal Pilipczuk, and Saket Saurabh.
\newblock {\em Parameterized Algorithms}.
\newblock Springer, 2015.
\newblock \href {http://dx.doi.org/10.1007/978-3-319-21275-3}
  {\path{doi:10.1007/978-3-319-21275-3}}.

\bibitem{DomLS14}
Michael Dom, Daniel Lokshtanov, and Saket Saurabh.
\newblock Kernelization lower bounds through colors and ids.
\newblock {\em {ACM} Trans. Algorithms}, 11(2):13:1--13:20, 2014.
\newblock \href {http://dx.doi.org/10.1145/2650261}
  {\path{doi:10.1145/2650261}}.

\bibitem{FellowsSPS15}
Michael~R. Fellows, U{\'{e}}verton dos Santos~Souza, F{\'{a}}bio Protti, and
  Maise~Dantas da~Silva.
\newblock Tractability and hardness of flood-filling games on trees.
\newblock {\em Theor. Comput. Sci.}, 576:102--116, 2015.
\newblock \href {http://dx.doi.org/10.1016/j.tcs.2015.02.008}
  {\path{doi:10.1016/j.tcs.2015.02.008}}.

\bibitem{FellowsPRSS17}
Michael~R. Fellows, F{\'{a}}bio Protti, Frances~A. Rosamond, Maise~Dantas
  da~Silva, and U{\'{e}}verton dos Santos~Souza.
\newblock Algorithms, kernels and lower bounds for the {Flood-It} game
  parameterized by the vertex cover number.
\newblock {\em Discrete Applied Mathematics}, 2017.
\newblock in press.
\newblock \href {http://dx.doi.org/10.1016/j.dam.2017.07.004}
  {\path{doi:10.1016/j.dam.2017.07.004}}.

\bibitem{FleischerW12}
Rudolf Fleischer and Gerhard~J. Woeginger.
\newblock An algorithmic analysis of the honey-bee game.
\newblock {\em Theor. Comput. Sci.}, 452:75--87, 2012.
\newblock \href {http://dx.doi.org/10.1016/j.tcs.2012.05.032}
  {\path{doi:10.1016/j.tcs.2012.05.032}}.

\bibitem{FlumG06}
J{\"{o}}rg Flum and Martin Grohe.
\newblock {\em Parameterized Complexity Theory}.
\newblock Texts in Theoretical Computer Science. An {EATCS} Series. Springer,
  2006.

\bibitem{FukuiOUUU12}
Hiroyuki Fukui, Yota Otachi, Ryuhei Uehara, Takeaki Uno, and Yushi Uno.
\newblock On complexity of flooding games on graphs with interval
  representations.
\newblock In {\em TJJCCGG 2012}, volume 8296 of {\em Lecture Notes in Computer
  Science}, pages 73--84, 2013.
\newblock \href {http://dx.doi.org/10.1007/978-3-642-45281-9_7}
  {\path{doi:10.1007/978-3-642-45281-9_7}}.

\bibitem{GajarskyLO13}
Jakub Gajarsk{\'{y}}, Michael Lampis, and Sebastian Ordyniak.
\newblock Parameterized algorithms for modular-width.
\newblock In {\em {IPEC} 2013}, volume 8246 of {\em Lecture Notes in Computer
  Science}, pages 163--176, 2013.
\newblock \href {http://dx.doi.org/10.1007/978-3-319-03898-8_15}
  {\path{doi:10.1007/978-3-319-03898-8_15}}.

\bibitem{HonKLLW15arxiv}
Wing{-}Kai Hon, Ton Kloks, Fu{-}Hong Liu, Hsiang~Hsuan Liu, and Hung{-}Lung
  Wang.
\newblock Flood-it on {AT}-free graphs.
\newblock {\em CoRR}, abs/1511.01806, 2015.
\newblock \href {http://arxiv.org/abs/1511.01806} {\path{arXiv:1511.01806}}.

\bibitem{LagoutteNT14}
Aur{\'{e}}lie Lagoutte, Mathilde Noual, and Eric Thierry.
\newblock Flooding games on graphs.
\newblock {\em Discrete Applied Mathematics}, 164:532--538, 2014.
\newblock \href {http://dx.doi.org/10.1016/j.dam.2013.09.024}
  {\path{doi:10.1016/j.dam.2013.09.024}}.

\bibitem{Lampis12}
Michael Lampis.
\newblock Algorithmic meta-theorems for restrictions of treewidth.
\newblock {\em Algorithmica}, 64(1):19--37, 2012.

\bibitem{McConnellS99}
Ross~M. McConnell and Jeremy~P. Spinrad.
\newblock Modular decomposition and transitive orientation.
\newblock {\em Discrete Mathematics}, 201(1-3):189--241, 1999.

\bibitem{MeeksS12}
Kitty Meeks and Alexander Scott.
\newblock The complexity of flood-filling games on graphs.
\newblock {\em Discrete Applied Mathematics}, 160(7--8):959--969, 2012.
\newblock \href {http://dx.doi.org/10.1016/j.dam.2011.09.001}
  {\path{doi:10.1016/j.dam.2011.09.001}}.

\bibitem{MeeksS13}
Kitty Meeks and Alexander Scott.
\newblock The complexity of \textsc{Free-Flood-It} on $2 \times n$ boards.
\newblock {\em Theor. Comput. Sci.}, 500:25--43, 2013.
\newblock \href {http://dx.doi.org/10.1016/j.tcs.2013.06.010}
  {\path{doi:10.1016/j.tcs.2013.06.010}}.

\bibitem{MeeksS14}
Kitty Meeks and Alexander Scott.
\newblock Spanning trees and the complexity of flood-filling games.
\newblock {\em Theory Comput. Syst.}, 54(4):731--753, 2014.
\newblock \href {http://dx.doi.org/10.1007/s00224-013-9482-z}
  {\path{doi:10.1007/s00224-013-9482-z}}.

\bibitem{MeeksV15arxiv}
Kitty Meeks and Dominik~K. Vu.
\newblock Extremal properties of flood-filling games.
\newblock {\em CoRR}, abs/1504.00596, 2015.
\newblock \href {http://arxiv.org/abs/1504.00596} {\path{arXiv:1504.00596}}.

\bibitem{Oum08}
Sang{-}il Oum.
\newblock Approximating rank-width and clique-width quickly.
\newblock {\em {ACM} Transactions on Algorithms}, 5, 2008.
\newblock Article No.\ 10.
\newblock \href {http://dx.doi.org/10.1145/1435375.1435385}
  {\path{doi:10.1145/1435375.1435385}}.

\bibitem{SouzaPS14}
U{\'{e}}verton {\noopsort{Souza}dos Santos Souza}, F{\'{a}}bio Protti, and
  Maise~Dantas da~Silva.
\newblock An algorithmic analysis of {Flood-it} and {Free-Flood-it} on graph
  powers.
\newblock {\em Discrete Mathematics {\&} Theoretical Computer Science},
  16(3):279--290, 2014.
\newblock URL: \url{http://dmtcs.episciences.org/2086}.

\bibitem{TedderCHP08}
Marc Tedder, Derek~G. Corneil, Michel Habib, and Christophe Paul.
\newblock Simpler linear-time modular decomposition via recursive factorizing
  permutations.
\newblock In {\em ICALP 2008 (1)}, volume 5125 of {\em Lecture Notes in
  Computer Science}, pages 634--645, 2008.
\newblock \href {http://dx.doi.org/10.1007/978-3-540-70575-8_52}
  {\path{doi:10.1007/978-3-540-70575-8_52}}.

\bibitem{Verbin2009blog-comment}
Elad Verbin.
\newblock Comment to ``{Is this game NP-Hard}? by {Sariel Har-Peled}.
\newblock \url{http://sarielhp.org/blog/?p=2005#comment-993}, 2009.
\newblock Accessed: 2018-01-18.

\end{thebibliography}



\end{document}